\documentclass[a4paper,USenglish,cleveref, autoref, thm-restate, numberwithinsect]{lipics-v2021}

\usepackage[utf8]{inputenc}
\usepackage[T1]{fontenc}
\usepackage{hyperref}

\usepackage{verbatim}
\usepackage{graphicx}
\usepackage{rotating}
\usepackage{booktabs}

\usepackage{url}
\usepackage{xcolor}
\definecolor{linkColor}{RGB}{156,78,13}
\usepackage{hyperref}
\hypersetup{
    colorlinks = true,
    allcolors = linkColor 
    }
\usepackage{xspace}

\usepackage{color}
\ifdefined\DEBUG{}
\newcommand{\ben}[1]{{\color{blue}{#1}}}
\newcommand{\bmp}[1]{{\color{purple}{#1}}}
\newcommand{\jjh}[1]{{\color{orange}{#1}}}
\newcommand{\BenSays}[1]{ \marginpar{  \begin{flushleft}\scriptsize \textbf{BB:}~{#1}\end{flushleft} } }

\def\rem#1{{\marginpar{\raggedright\scriptsize #1}}}
\newcommand{\bmpr}[1]{\rem{\textcolor{purple}{\(\bullet \) #1}}}

\newcommand{\JariSays}[1]{ \marginpar{  \begin{flushleft}\scriptsize \textbf{J:}~{#1}\end{flushleft} } }

\else
\newcommand{\ben}[1]{#1}
\newcommand{\bmp}[1]{#1}
\newcommand{\jjh}[1]{#1}
\newcommand{\bmpr}[1]{}
\newcommand{\BenSays}[1]{}
\newcommand{\JariSays}[1]{}
\fi

\usepackage{quiver}
\usepackage{pgf,tikz,environ}
\usetikzlibrary{arrows}
\usetikzlibrary{cd}
\usetikzlibrary{positioning}
\usetikzlibrary{fadings}
\usetikzlibrary{decorations.pathreplacing}
\usetikzlibrary{decorations.pathmorphing}
\tikzset{snake it/.style={decorate, decoration=snake}}

\usepackage{mathtools}

\newtheorem*{theorem*}{Theorem}
\newtheorem*{lemma*}{Lemma}

\usepackage{enumitem}
\usepackage{amsmath}
\usepackage{framed} 

\newcommand{\Oh}{\mathcal{O}}
\newcommand{\C}{\mathcal{C}}
\newcommand{\forb}{\mathsf{forb}}

\newcommand{\flower}[2]{(#1,#2)-flower}
\newcommand{\flowernumber}[3]{\Gamma_{#1}(#2, #3)}

\newcommand{\typesetproblem}[1]{\textsc{#1}}

\DeclareMathOperator{\fpt}{\typesetproblem{FPT}}

\DeclareMathOperator{\opt}{\textsc{OPT}}

\newcommand{\psdetectshort}[1]{\typesetproblem{$c$-Essential detection}}
\newcommand{\psdetect}[1]{\typesetproblem{$#1$-Essential detection for $\Pi$}}
\newcommand{\psdetectFor}[2]{\typesetproblem{$#1$-Essential detection for #2}}
\newcommand{\CDeletion}{\textsc{$\C$-Deletion}\xspace}

\DeclareMathOperator{\ds}{\typesetproblem{Dominating Set}}
\DeclareMathOperator{\wheelfreeedge}{\typesetproblem{WheelDelE}}
\DeclareMathOperator{\wheelfree}{\typesetproblem{WheelDel}}
\DeclareMathOperator{\hittingset}{\typesetproblem{HittingSet}}
\DeclareMathOperator{\oct}{\typesetproblem{OCT}}
\DeclareMathOperator{\doct}{\typesetproblem{DOCT}}
\DeclareMathOperator{\fvs}{\typesetproblem{FVS}}
\DeclareMathOperator{\dfvs}{\typesetproblem{DFVS}}
\DeclareMathOperator{\chordDel}{\typesetproblem{Chordal Deletion}}

\newcommand{\hole}{\ensuremath{\mathsf{hole}}\xspace}
\DeclareMathOperator{\perfectDel}{\typesetproblem{PerfDel}}

\newcommand{\deflp}[3]{
	\vspace{1mm}
	\noindent\fbox{
		\begin{minipage}{0.96\textwidth}
			\begin{tabular*}{\textwidth}{@{\extracolsep{\fill}}lr} \textsc{#1} &  \\ \end{tabular*}
			{\bf{Objective:}} #2 \\
			{\bf{Subject to:}} #3
		\end{minipage}
	}
	\vspace{1mm}
}

\newcommand{\deftask}[3]{
	\vspace{1mm}
	\noindent\fbox{
		\begin{minipage}{0.96\textwidth}
			\begin{tabular*}{\textwidth}{@{\extracolsep{\fill}}lr} \textsc{#1} &  \\ \end{tabular*}
			{\bf{Input:}} #2 \\
			{\bf{Task:}} #3
		\end{minipage}
	}
	\vspace{1mm}
}

\title{Search-Space Reduction via Essential Vertices}

\author{Benjamin Merlin Bumpus}{Eindhoven University of Technology, the Netherlands}{b.bumpus@tue.nl}{https://orcid.org/0000-0002-8686-2319}{}

\author{Bart M.P. Jansen}{Eindhoven University of Technology, the Netherlands}{b.m.p.jansen@tue.nl}{https://orcid.org/0000-0001-8204-1268}{}

\author{Jari J.H. de Kroon}{Eindhoven University of Technology, the Netherlands}{j.j.h.d.kroon@tue.nl}{https://orcid.org/0000-0003-3328-9712}{}

\funding{Funded by the European Research Council (ERC) under the European Union's Horizon 2020 research and innovation programme (grant agreement No 803421, ReduceSearch). \flag{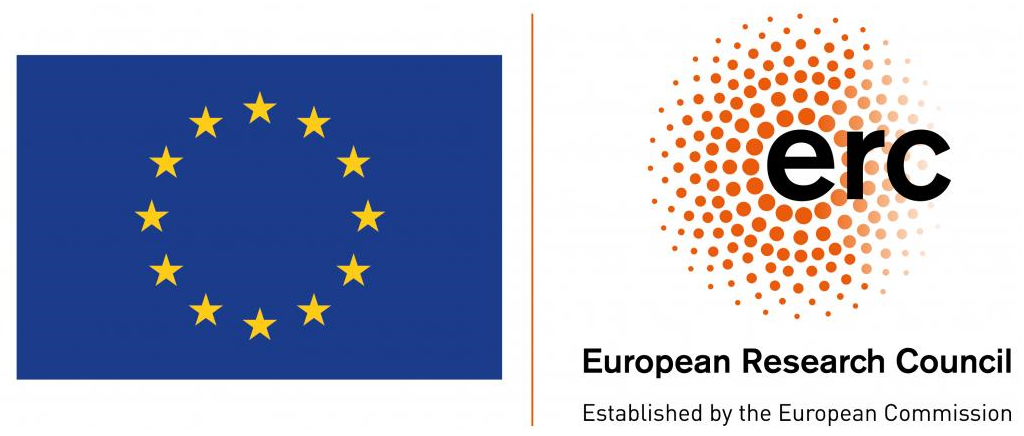}}

\authorrunning{B.M. Bumpus, B.M.P. Jansen, and J.J.H. de Kroon} 
\Copyright{Benjamin M. Bumpus, Bart M.P. Jansen, and Jari J.H. de Kroon}

\ccsdesc[500]{Theory of computation~Graph algorithms analysis}
\ccsdesc[500]{Theory of computation~Packing and covering problems}
\ccsdesc[300]{Theory of computation~Linear programming}
\ccsdesc[500]{Theory of computation~Fixed parameter tractability}

\keywords{fixed-parameter tractability, essential vertices, covering versus packing}

\hideLIPIcs
\nolinenumbers

\begin{document}

\maketitle

\begin{abstract}
We investigate preprocessing for vertex-subset problems on graphs. While the notion of kernelization, originating in parameterized complexity theory, is a formalization of provably effective preprocessing aimed at reducing the total instance size, our focus is on finding a non-empty vertex set that belongs to an optimal solution. This decreases the size of the remaining part of the solution which still has to be found, and therefore shrinks the search space of fixed-parameter tractable algorithms for parameterizations based on the solution size. We introduce the notion of a $c$-essential vertex as one that is contained in all $c$-approximate solutions. For several classic combinatorial problems such as \textsc{Odd Cycle Transversal} and \textsc{Directed Feedback Vertex Set}, we show that under mild conditions a polynomial-time preprocessing algorithm can find a subset of an optimal solution that contains all 2-essential vertices, by exploiting packing/covering duality. This leads to FPT algorithms to solve these problems where the exponential term in the running time depends only on the number of \emph{non-essential} vertices in the solution.
\end{abstract}

\section{Introduction}\label{sec:intro}
\subparagraph*{Background and motivation}
Due to the enormous potential of preprocessing to speed up the algorithmic solution to combinatorial problems~\cite{AchterbergBGRW16,AkibaI16,AlberBN06,Weihe98,Weihe01}, a large body of work in theoretical computer science is devoted to understanding its power and limitations. Using the notion of \emph{kernelization}~\cite{Bodlaender14,Fellows06,FominLSM19,GuoN07,LokshtanovMS12,Kratsch14} from parameterized complexity~\cite{CyganFKLMPPS15,DowneyF12} it is possible to formulate a guarantee on the size of the instance after preprocessing based on the  parameter of the original instance. Under this model, a good preprocessing algorithm is a \emph{kernelization algorithm}: given a parameterized instance~$(x,k)$, it outputs an equivalent instance~$(x',k')$ of the same decision problem such that~$|x'| + k' \leq f(k)$ for some function~$f$ that bounds the size of the kernel. Research into kernelization led to deep algorithmic insights, including connections to protrusions and finite-state properties~\cite{BodlaenderFLPST09}, well-quasi ordering~\cite{FominLMS12}, and matroids~\cite{KratschW20}; these positive results were complemented by impossibility results~\cite{DellM14,Drucker15,KratschW20} delineating the boundaries of tractability.

Results on kernelization led to profound insights into the limitations of polynomial-time data compression for NP-hard problems. However, as recently advocated~\cite{DonkersJ21}, the definition of kernelization only gives guarantees on the \emph{size} of the instance after preprocessing, which does not directly correspond to the running time of subsequently applied algorithms. If the preprocessed instance is not solved by brute force, but via a fixed-parameter tractable algorithm whose running time is of the form~$f(k)\cdot n^{\Oh(1)}$, then the exponential dependence in the running time is on the value of the \emph{parameter}~$k$, which is not guaranteed to decrease via kernelization. In fact, if $P \neq NP$ then no polynomial-time preprocessing algorithm can guarantee to always decrease the parameter of an NP-hard fixed-parameter tractable problem, as iterating the preprocessing algorithm would lead to its solution in polynomial time. In this work, we develop a new analysis of preprocessing aimed at reducing the search space of the follow-up algorithm. We apply this framework to combinatorial optimization problems on graphs, whose goal is to find a minimum vertex-subset satisfying certain properties. The main idea behind the framework is to define formally what it means for a vertex to be \emph{essential} for making reasonable solutions to the problem, and to prove that an efficient preprocessing algorithm can detect a subset of an optimal solution that contains all such essential vertices.

Before stating our results, we introduce and motivate the model. We consider vertex-subset minimization problems on (possibly directed) graphs, in which the goal is to find a minimum vertex subset having a certain property. Examples of the problems we study include \textsc{Vertex Cover}, \textsc{Odd Cycle Transversal}, and \textsc{Dominating Set}. The analysis of such minimization problems, parameterized by the size of the solution, has played an important role in the literature (cf.~\cite{ChenKX10,EibenGHK21,JansenK021,ReedSV04,karthik2017parameterized}). Our starting point is the thesis that a good preprocessing algorithm should reduce the \emph{search space}. Since many graph problems are known to be fixed-parameter tractable when parameterized by the size of the solution, we can reduce the search space of these FPT algorithms by finding one or more vertices which are part of an optimal solution, thereby decreasing the size of the solution still to be found in the reduced instance (i.e.~the parameter value).

Since in general no polynomial-time algorithm can guarantee to identify at least one vertex that belongs to an optimal solution, the guarantee of the effectiveness of the preprocessing algorithm should be stated in a more subtle way. When solving problems by hand, one sometimes finds that certain vertices~$v$ are easily identified to belong to an optimal solution, as avoiding them would force the solution to contain a prohibitively large number of alternative vertices and therefore be suboptimal. \bmp{Can an efficient preprocessing algorithm identify those vertices that cannot be avoided when making an optimal solution?}

Since many NP-hard problems remain hard even when there is a unique solution~\cite{ValiantV86}, this turns out to be too much to ask as it would allow instances with unique solutions to be solved in polynomial time, which leads to~$NP = RP$. We therefore have to relax the requirements on the preprocessing guarantee slightly, as follows. For an instance of a vertex-subset minimization problem~$\Pi$ on a graph~$G$, we denote the minimum size of a solution on~$G$ by~$\opt_\Pi(G)$. For a fixed~$c \in \mathbb{R}_{\geq 1}$, we say a vertex~$v \in V(G)$ is \emph{$c$-essential} for~$\Pi$ on~$G$ if all feasible solutions~$S\subseteq V(G)$ for~$\Pi$ whose total size is at most~$c \cdot \opt_\Pi(G)$ contain~$v$. Based on this notion, we can ask ourselves: can an efficient preprocessing algorithm identify part of an optimal solution if there is at least one $c$-essential vertex?

Phrased in this way, the algorithmic task becomes more tractable. For example, for the \textsc{Vertex Cover} problem, selecting all vertices that receive the value~$1$ in an optimal half-integral solution to the LP-relaxation results in a set~$S$ which is contained in some optimal solution (by the Nemhauser-Trotter theorem~\cite{NemhauserT75}, cf~\cite[\S 2.5]{CyganFKLMPPS15}), and at the same time includes all $2$-essential vertices: any vertex~$v \notin S$ only has neighbors of value~$\frac{1}{2}$ and~$1$, which implies that the set~$X$ of vertices other than~$v$ whose value in the LP relaxation is at least~$\frac{1}{2}$, forms a feasible solution which avoids~$v$. Its cardinality is at most twice the cost of the LP relaxation and therefore~$X$ is a 2-approximation. Hence~$S$ contains all $2$-essential vertices.

This example shows that a preprocessing step that detects $c$-essential vertices without any additional information is sometimes possible. However, to be able to extend the scope of our results also to problems which \emph{do not} have polynomial-time constant-factor approximations, we slightly relax the requirements on the preprocessing algorithm as follows. Let~$\Pi$ be a minimization problem on graphs whose solutions are vertex subsets and let~$c \in \mathbb{R}_{\geq 1}$.

\deftask
{\psdetect{c}}
{\bmp{A graph $G$ and integer~$k$.}}
{Find a vertex set~$S \subseteq V(G)$ such that:
    \begin{enumerate}[label=\textbf{G\arabic*}]
        \item if~$\opt_\Pi(G) \leq k$, then there is an optimal solution in~$G$ containing all of~$S$, and \label{prop:c-essential-sup-det-double-guarantee-1} 
        \item if~$\opt_\Pi(G) = k$, then~$S$ contains all $c$-essential vertices.\label{prop:c-essential-sup-det-double-guarantee-2} 
    \end{enumerate}}

In this model, the preprocessing task is facilitated by supplying an additional integer~$k$ in the input. The correctness properties of the output~$S$ are formulated in terms of~$k$. If~$\opt_\Pi(G) \leq k$, then the set~$S$ is required to be part of an optimal solution. The upper bound on~$\opt_\Pi(G)$ is useful to the algorithm: whenever the algorithm establishes that avoiding~$v$ would incur a cost of more than~$k$, it is safe to add~$v$ to~$S$. If~$\opt_\Pi(G) = k$, then the algorithm should guarantee that~$S$ contains all $c$-essential vertices. Knowing a lower bound on~$\opt_\Pi(G)$ is useful to the algorithm in case it can establish that any optimal solution containing~$v$ can be transformed into one avoiding~$v$ whose cost is~$(c-1) \cdot k$ larger, which yields a $c$-approximation if~$(c-1) \cdot k \leq (c-1) \opt_\Pi(G)$. Hence vertices for which such a replacement exists are not $c$-essential and may safely be left out of~$S$.

\subparagraph*{Results} We present polynomial-time algorithms for \psdetect{c} for a range of vertex-deletion problems~$\Pi$ and small values of~$c$; typically~$c \in \{2,3\}$. Example applications include \textsc{Vertex Cover} and \textsc{Feedback Vertex Set}, and also \textsc{Chordal Vertex Deletion} (for which no $O(1)$-approximation is known), \textsc{Odd Cycle Transversal} (for which no $O(1)$-approximation exists, assuming the Unique Games Conjecture~\cite{unique-games-Khot, wernicke2014algorithmic}), and even \textsc{Directed Odd Cycle Transversal} (which is $W[1]$-hard parameterized by solution size~\cite{LokshtanovR0Z20}).

The model of \psdetect{c} is chosen such that the detection algorithms whose correctness is formulated with respect to the value of~$k$, can be used as a preprocessing step to optimally solve vertex-subset problems without any knowledge of the optimum. \bmp{Let~$\mathcal{E}^\Pi_c(G)$ denote the set of $c$-essential vertices in~$G$, which is well-defined since all optimal solutions contain all $c$-essential vertices.} By using a preprocessing step that detects a superset of the $c$-essential vertices in the solution, we can effectively improve the running-time guarantee for FPT algorithms parameterized by solution size from~$f(\opt_\Pi(G)) \cdot |V(G)|^{\Oh(1)}$, to~$f(\opt_\Pi(G) - |\mathcal{E}^\Pi_c(G)|) \cdot |V(G)|^{\Oh(1)}$. This leads to the following results.

\begin{theorem} \label{thm:fpt:main}
For each problem~$\Pi$ with coefficient~$c$ and parameter dependence~$f$ listed in Table~\ref{table:summary} that is not $W[1]$-hard, there is an algorithm that, given a graph~$G$, outputs an optimal solution in time~$f(\ell)\cdot |V(G)|^{\Oh(1)}$, where~$\ell := \opt_\Pi(G) - |\mathcal{E}^\Pi_c(G)|$ is the number of vertices in an optimal solution which are not $c$-essential.
\end{theorem}

Hence the running time for solving these problems does not depend on the total size  of an optimal solution, only on the part of the solution that does not consist of $c$-essential vertices. The theorem effectively shows that by employing \psdetect{c} as a preprocessing step, the size of the search space no longer depends on the total solution size but only on its non-essential vertices.

We also prove limitations to this approach. Assuming~$\fpt \neq W[1]$, for \textsc{Dominating Set}, \textsc{Perfect Deletion} (in which the goal is to obtain a perfect graph by vertex deletions) and \textsc{Wheel-Free Deletion}, there is no polynomial-time algorithm for \psdetectshort{c} for any~$c \in \mathbb{R}_{\geq 1}$. In fact, we can even rule out such algorithms running in time~$f(k) \cdot |V(G)|^{\Oh(1)}$. These results are based on FPT-inapproximability results for \textsc{Dominating Set}~\cite{karthik2017parameterized} and existing reductions~\cite{HeggernesHJKV13,Lokshtanov-wheel-free-2008} to the mentioned vertex-deletion problems.

\begin{table}
\centering
\caption{For each  problem~$\Pi$, there is a polynomial-time algorithm for \psdetectshort{} for the stated value of~$c$. Combined with the state of the art algorithm for the natural parameterization, this leads to an algorithm solving the problem in time~$f(\ell) \cdot |V(G)|^{\Oh(1)}$ where~$\ell = \opt_\Pi(G) - |\mathcal{E}^\Pi_c(G)|$.} \label{table:summary}
\begin{tabular}{@{}llrl} \toprule
Problem & $c$ & $f(\ell)$ & Reference \\ \midrule
\textsc{Vertex Cover} & $2$ & $1.2738^\ell$ &  \cite{ChenKX10} \\
\textsc{Feedback Vertex Set} & $2$ & $2.7^\ell$ & \cite{LiN20} \\
\textsc{Directed Feedback Vertex Set} & $2$ &  $4^\ell \cdot \ell!$ & \cite{ChenLLOR08} \\
\textsc{Odd Cycle Transversal} & $2$ & $2.3146^\ell$ & \cite{LokshtanovNRRS14} \\
\textsc{Directed Odd Cycle Transversal} & 3 & \bmp{W[1]-hard} & \cite{LokshtanovR0Z20} \\
\textsc{Chordal Vertex Deletion} & 13 & $2^{\Oh(\ell \log \ell)}$ & \cite{CaoM16} \\ \bottomrule
\end{tabular}
\end{table}

\subparagraph*{Techniques} 
The main work lies in the algorithms for \psdetectshort{c}, which are all based on covering/packing duality for forbidden induced subgraphs to certain graph classes, or variations thereof in terms of (integer) solutions to certain linear programs and their (integer) duals. To understand the relation between detecting essential vertices and covering/packing duality, consider the \textsc{Odd Cycle Transversal} problem (OCT). Following the argumentation for the classic Erd\H{o}s-Pósa theorem~\cite{ErdosP65}, in general there is no constant~$c$ such that any graph either has an odd cycle transversal of size~$c \cdot k$, or a packing of $k$ vertex-disjoint odd cycles. 
However, we show that a linear packing/covering relation holds in the following slightly different setting. If~$G - v$ is bipartite (so all odd cycles of~$G$ intersect~$v$), then the minimum size of an OCT avoiding~$v$ equals the maximum cardinality of a packing of odd cycles which pairwise intersect only at~$v$. We can leverage this statement to prove that any vertex~$v$ which is not at the center of a flower (see Definition~\ref{def:flower}) of more than~$\opt_{\mathrm{oct}}(G)$ odd cycles, is not $2$-essential: for any optimal solution~$X$ containing~$v$, the graph~$G' := G - (X \setminus \{v\})$ becomes bipartite after removal of~$v$ and by assumption does not contain a packing of more than~$\opt_{\mathrm{oct}}(G)$ odd cycles pairwise intersecting at~$v$. So by covering/packing duality on~$G'$, it has an OCT~$X'$ of size at most~$\opt_{\mathrm{oct}}(G)$ avoiding~$v$, so that~$(X \setminus \{v\}) \cup X'$ is a 2-approximation in~$G$ which avoids~$v$, showing that~$v$ is not 2-essential. Since~$v$ is clearly contained in an optimal solution whenever there is a flower of more than~$\opt_{\mathrm{oct}}(G)$ odd cycles centered at~$v$, this yields a method for \psdetectshort{2} when using \bmp{a known reduction~\cite{DBLP:journals/jct/GeelenGRSV09} to \textsc{Maximum Matching}} to test for a flower of odd cycles.

\subparagraph*{Organization} After presenting preliminaries in Section~\ref{sec:prelims}, we give algorithms to detect essential vertices based on covering/packing duality in Section~\ref{sec:positive_packcover} and based on integrality gaps in Section~\ref{sec:positive_lp}. In Section~\ref{sec:fpt} we show how these detection subroutines can be used to improve the parameter dependence of FPT algorithms parameterized by solution size. The lower bounds are presented in Section~\ref{sec:hardness}. The investigation of $c$-essential vertices has close connections to the area of perturbation stability, which we briefly explore in Section~\ref{sec:stability}. We conclude in Section~\ref{sec:conclusion}.

\section{Preliminaries}\label{sec:prelims} 
We consider finite simple graphs, some of which are directed. We use standard notation for graph algorithms; any terms not defined here can be found in the textbook by Cygan et al.~\cite{CyganFKLMPPS15}.
\bmp{We consider vertex-deletion problems on graphs. For a graph class~$\C$, a $\C$-modulator in a graph~$G$ is a vertex set~$S \subseteq V(G)$ such that~$G - S \in \C$. The minimum size of a $\C$-modulator in~$G$ is denoted~$\opt_\C(G)$. The corresponding minimization problem is defined as follows.}

\deftask{$\C$-Deletion}
{\bmp{A graph $G$.}}
{Find a minimum-size vertex-subset $S \subseteq V(G)$ such that $G - S \in \C$.}

\ben{Throughout this paper we consider hereditary graph classes $\C$. These} can be characterized by a (possibly infinite) set of forbidden induced subgraphs denoted $\forb(\C)$. The \CDeletion problem is equivalent to finding a minimum set $S \subseteq V(G)$ such that no induced subgraph of $G-S$ is isomorphic to a graph in $\forb(\C)$. We say that such a set $S$ hits all graphs from $\forb(\C)$ in $G$. \jjh{For these classes the vertex set~$V(G)$ is a trivial $\C$-modulator (since the empty graph is in all hereditary classes), which ensures that the task of finding a smallest modulator is always well-defined.}

A graph is perfect if for every induced subgraph the size of a largest clique is equivalent to its \ben{chromatic} number. Equivalently, a graph is perfect if it has no induced cycle of odd length at least five or its edge complement (cf.~\cite{Golumbic2002}). A graph is chordal if it has no induced cycle of length at least four. A graph is bipartite if its vertex set can be partitioned into two independent sets, or equivalently, it does not contain an odd-length cycle. Given a graph $G$ and a set $T \subseteq V(G)$, a $T$-path is a path with at least one edge with both endpoints contained in~$T$.
A $T$-path is odd if it has an odd number of edges. 
For $u,v \in V(G)$, a $(u,v)$-separator is a set $S \subseteq V(G) \setminus \{u,v\}$ that disconnects $u$ from $v$. If $G$ is a directed graph, then in $G-S$ there is no directed path from $u$ to $v$ instead.

\section{Positive results via Packing Covering}\label{sec:positive_packcover}
In this section we provide polynomial-time algorithms for $\psdetect{c}$ for various problems $\Pi$. \bmp{The case for the \textsc{Vertex Cover} problem was given in Section~\ref{sec:intro}. The results in this section are all based on packing/covering duality (cf.~\cite{ChudnovskyGGGLS06}, \cite[\S 73]{Schrijver03}). Towards this end, we}  generalize the notion of \emph{flowers}\bmp{, which played a key role in kernelization for \textsc{Feedback Vertex Set}~\cite{BodlaenderD10}. While flowers were originally formulated as systems of cycles (forbidden structures for \textsc{Feedback Vertex Set}) pairwise intersecting in a single common vertex, we generalize the notion to near-packings of arbitrary structures here.}

\begin{definition}\label{def:flower}
Let $\mathcal{F}$ be a set of graphs. \bmp{For a graph~$G$ and~$v \in V(G)$,} a \emph{\flower{$v$}{$\mathcal{F}$} with $p$ petals} in $G$ is a set $\{C_1, C_2, \dots , C_p\}$ of \bmp{induced} subgraphs of~$G$ such that each~$C_i$ (with $i \in [p]$) is isomorphic to some member of $\mathcal{F}$and such that $V(C_i) \cap V(C_j) = \{v\}$ for all distinct $i,j \in [p]$.
The \emph{$\mathcal{F}$-flower number} of a vertex $v \in V(G)$, denoted $\flowernumber{\mathcal{F}}{G}{v}$, is the largest integer $p$ for which there is a \flower{$v$}{$\mathcal{F}$} in $G$ with $p$ petals.
\end{definition}

\bmp{We show a general theorem for finding 2-essential vertices for $\C$-deletion if a maximum \flower{$v$}{$\forb(\C)$} can be computed in polynomial-time. It applies to those classes~$\C$ where graphs with~$G - v \in \C$ obey a min-max relation between $\C$-modulators avoiding~$v$ and packings of forbidden induced subgraphs intersecting only at~$v$.}

\begin{theorem}\label{thm:maxflower+minmax_on_C}
Let~$\C$ be a hereditary graph class such that, for any graph $G$ and vertex $v \in V(G)$ with $G-v \in \C$, \bmp{the minimum size of a $\C$-modulator avoiding~$v$ in~$G$ equals~$\flowernumber{\forb(\C)}{G}{v}$.}
Suppose \bmp{there exists} a polynomial-time algorithm $\mathcal{A}$ that, given a graph $G$ and vertex $v \in V(G)$, computes $\flowernumber{\forb(\C)}{G}{v}$. Then there is a polynomial-time algorithm that solves $\psdetectFor{2}{\CDeletion}$.
\end{theorem}
\begin{proof}

Apply algorithm~$\mathcal{A}$ to each vertex $v \in V(G)$ and let $S$ be the set of vertices for which it finds that $\flowernumber{\forb(\C)}{G}{v} > k$. We argue that Properties~\ref{prop:c-essential-sup-det-double-guarantee-1} and~\ref{prop:c-essential-sup-det-double-guarantee-2} are satisfied. If $\opt_\C(G) \leq k$, then every vertex in $S$ is contained in every optimal solution for $G$ \bmp{since a size-$k$ solution cannot hit a flower of~$k+1$ petals from~$\forb(\C)$ without using~$v$}. Therefore Property~\ref{prop:c-essential-sup-det-double-guarantee-1} is satisfied. Next suppose that $\opt_\C(G) = k$ and let $X$ be an optimal solution. We argue that each vertex $v \notin S$ is not 2-essential. Clearly this holds for any vertex $v \notin X$, so suppose that $v \in X$. Note that for every vertex $v \notin S$ we have $\flowernumber{\forb(\C)}{G}{v} \leq k$, which implies that $\flowernumber{\forb(\C)}{G'}{v} \leq k$ where $G' := G-(X \setminus \{v\})$. Note that since $G'-v \in \C$, by assumption there exists a $\C$-modulator $X' \subseteq V(G')\setminus\{v\}$ in $G'$ of size $\flowernumber{\forb(\C)}{G'}{v} \leq k$. Observe that $(X \setminus \{v\}) \cup X'$ is a $\C$-modulator in $G$ of size at most $2k$ that avoids $v$ and therefore $v$ is not 2-essential.
\end{proof}

\ben{Theorem~\ref{thm:maxflower+minmax_on_C} allows us to conclude the following result for \textsc{Feedback Vertex Set} ($\fvs$) and its directed variant ($\dfvs$) using Gallai's theorem and Menger's theorem, respectively.}

\begin{lemma}\label{lem:packcover_applications}
There are polynomial-time algorithms for \psdetectFor{2}{$\Pi$} for $\Pi \in \{\fvs,\dfvs\}$. 
\end{lemma}
\begin{proof}
The $\fvs$ problem corresponds to \CDeletion where $\C$ is the set of acyclic graphs and $\forb(\C)$ is the set of cycles. We argue that the preconditions of Theorem~\ref{thm:maxflower+minmax_on_C} hold. The existence of a polynomial-time algorithm that computes $\flowernumber{\forb(\C)}{G}{v}$ is implied by Gallai's theorem (as written in Cygan et al.~\cite[Theorem 9.2]{CyganFKLMPPS15}), \bmp{since a flower of cycles through~$v$ corresponds to a collection of paths connecting pairs of distinct vertices of~$N_G(v)$ to each other in the graph~$G-v$}. By applying \bmp{Gallai's theorem} to $G-v$ with $T = N_G(v)$, we get that $\flowernumber{\forb(\C)}{G}{v}$ is the largest integer $k$ such that the outcome of Gallai's theorem applied with the integer $k-1$ is a family of $k$ pairwise vertex-disjoint $T$-paths in $G-v$.
For the first precondition of Theorem~\ref{thm:maxflower+minmax_on_C}, we argue that \bmp{when~$G-v$ is acyclic,} the maximum number of petals in a \flower{$v$}{$\forb(\C)$} is equivalent to the minimum cardinality of a set $S$ that hits all cycles through $v$.

\begin{claim}[{\cite[page 67]{Berge89}}]
Let $T$ be a tree and let $\mathcal{F}$ be a collection of connected subgraphs of $T$. The maximum size packing of vertex-disjoint members of $\mathcal{F}$ equals the minimum size of a vertex set intersecting all of $\mathcal{F}$. 
\end{claim}

By noting that any cycle through $v$ in $G$ corresponds to a path between two neighbors in a connected component of $G-v$ which is a tree, the claim above directly implies the desired result. It follows there is a polynomial-time algorithm for \psdetectFor{2}{$\fvs$} by Theorem~\ref{thm:maxflower+minmax_on_C}.

The $\dfvs$ problem corresponds to \CDeletion where $\C$ is the set of directed acyclic graphs and $\forb(\C)$ is the set of directed cycles.
The preconditions of Theorem~\ref{thm:maxflower+minmax_on_C} follow from known results, see for instance the work of Fleischer et al.~\cite{DBLP:conf/esa/FleischerWY09}, we repeat it for completeness.
Consider any directed graph $D$ and vertex $v \in V(D)$.
Obtain a graph $D'$ by splitting $v$ into $v_i$ and $v_o$. Attach every incoming arc of $v$ to $v_i$ and every outgoing arc to $v_o$. Compute a minimum $(v_i,v_o)$-separator $S \subseteq V(D') \setminus \{v_i,v_o\}$ in $D'$. 
By Menger's theorem, its size is equivalent to the maximum number of internally \bmp{vertex-}disjoint paths from $v_i$ to $v_o$. Since any $(v_i,v_o)$-path in $D'$ is a directed cycle containing $v$ in $D$, it follows that $\flowernumber{\forb(\C)}{D}{v}$ can be computed in polynomial time. \bmp{Finally,} suppose that $D-v$ is a directed acyclic graph. Note that then $S$ is a $\C$-modulator avoiding $v$ of size $\flowernumber{\forb(\C)}{D}{v}$. By Theorem~\ref{thm:maxflower+minmax_on_C} it follows that there is a polynomial-time algorithm for \psdetectFor{2}{$\dfvs$}.
\end{proof}

\subsection{Odd cycle transversal}
Next, we consider \textsc{Odd Cycle Transversal} ($\oct$), which corresponds to \CDeletion where $\C$ is the class of bipartite graphs and $\forb(C)$ is consists of all odd cycles. In order to apply Theorem~\ref{thm:maxflower+minmax_on_C} to $\oct$, we first argue that the class of bipartite graphs satisfies the preconditions. The proof is similar to that of Geelen et al.~\cite[Lemma 11]{DBLP:journals/jct/GeelenGRSV09} who reduce the problem of packing odd cycles containing $v$ to a matching problem. \ben{We note that, although their result \bmp{can be used to obtain} a 3-essential detection algorithm, we will show (Lemma~\ref{lemma:oct-2-essential}) how to efficiently detect 2-essential vertices as well.} \jjh{If the graph resulting from their construction has a large matching}, then there is a large \flower{$v$}{$\forb(\C)$}. If on the other hand there is no large matching, then the Tutte-Berge formula is used to obtain a set of size~$2k$ that hits all the odd cycles passing through $v$. We show that if the graph $G-v$ is bipartite instead, then this second argument can be improved to obtain a hitting set of size~$k$ by noting that the lack of a large matching implies that there is a small vertex cover due to K\H{o}nig's theorem. This is below, \bmp{using the viewpoint that odd cycles in~$G$ correspond to odd $T$-paths in~$G-v$ for $T=N_G(v)$.}

\begin{lemma}\label{lem:oddflower+bipartite_packing_covering_equality}
For any \bmp{undirected} graph~$G$ and set~$T \subseteq V(G)$, a maximum packing of odd $T$-paths can be computed in polynomial time. Moreover, if~$G$ is bipartite then the cardinality of a maximum packing of odd $T$-paths is equal to the minimum size of a vertex set which intersects all odd $T$-paths.
\end{lemma}
\begin{proof}

We reduce to matching as in~\cite[Lemma 11]{DBLP:journals/jct/GeelenGRSV09}. Construct a graph $H$ as follows. For each $v \in V(G) \setminus T$, let $v' \notin V(G)$ be a copy of $v$. Let $V(H) = V(G) \cup \{v' \mid v \in V(G) \setminus T\}$ and $E(H) = E(G) \cup \{u'v' \mid uv \in E(G - T)\} \cup \{vv' \mid v \in V(G) \setminus T\}$. Note that the graph $H$ consists of the disjoint union of $G$ and a copy of $G-T$, with an added edge between $v \in V(G) \setminus T$ and its copy $v'$. Geelen et al.~\cite{DBLP:journals/jct/GeelenGRSV09} mention that there is a 1-1 correspondence between odd $T$-paths in $G$ and certain augmenting paths in $H$. For completeness we give a self-contained argument.

\begin{claim}\label{claim:oddTpaths_to_matching}
Graph $G$ contains $k$ vertex-disjoint odd $T$-paths if and only if $H$ has a matching $M$ of size $|V(G)\setminus T| + k$. Furthermore, given a matching $M$ in $H$ of size $|V(G)\setminus T| + k$ we can compute a set of $k$ vertex-disjoint odd $T$-paths in polynomial time.
\end{claim}
\begin{claimproof}
($\Rightarrow$) Let $\mathcal{P} = (P_1,\dots,P_k)$ be a set of $k$ vertex-disjoint odd $T$-paths in $G$. Consider a path $P = (v_1, \dots, v_{2\ell}) \in \mathcal{P}$, where $\ell \geq 1$. First note that we can assume that $V(P) \cap T = \{v_1, v_{2\ell}\}$, since if $v_i \in T$ for some $1 < i < 2\ell$, then either $(v_1,\dots,v_i)$ or $(v_i,\dots,v_{2\ell})$ is an odd $T$-path and we can update $\mathcal{P}$ accordingly. Construct a matching $M$ in $H$ as follows. For any path $P = (v_1, \dots, v_{2\ell}) \in \mathcal{P}$, add the edges $v_1v_2,v_2'v_3',\dots,v_{2\ell-1}v_{2\ell}$, alternating between original vertices and copy vertices. This is possible as $P$ is of odd length. For any vertex in $u \in V(G) \setminus T$ that is not contained in an odd $T$-path, we add $uu'$ to $M$. Observe that at least $2|V(G)\setminus T| + 2k$ vertices are matched, therefore $|M| \geq |V(G) \setminus T| + k$ as desired. 

($\Leftarrow$)
Let $M$ be a matching of size $|V(G)\setminus T| + k$. If $M$ contains both $uv$ and $u'v'$ for $u,v \in V(G) \setminus T$, then update $M$ by removing them and inserting $uu'$ and $vv'$ instead. If for $v \in V(G) \setminus T$ only one of $v$ and $v'$ is matched, and it is not matched to its copy, then match it to its copy instead. Afterwards let $E' := \{uv \in E(G) \mid uv \in M \vee u'v' \in M\}$. Observe that in $G[E']$, each vertex in $V(G) \setminus T$ \bmp{has} degree 0 or 2. For each $v \in V(G) \setminus T$ such that $v$ has degree 0 in $G[E']$, add $vv'$ to $M$ if it is not in already. Note that all vertices of $H-T$ are matched. It follows that at least $2k$ vertices in $T$ are matched by $M$ and they have degree 1 in $G[E']$. Observe that $G[E']$ is a collection of paths and cycles with all degree-1 vertices in $T$. We get that there are $k$ $T$-paths in $G$ that are of odd length by construction (every even numbered edge in $G[E']$ originated from the copy part of $H$). Note that we can find them in polynomial time. 
\end{claimproof}

Since a maximum matching can be computed in polynomial time, by the claim above we get that a maximum packing of vertex-disjoint odd $T$-paths can be computed in polynomial time. Next we prove the second part of the statement.

\begin{claim} \label{claim:bipartite}
If $G$ is bipartite and a maximum matching $M$ in $H$ has size $|V(G)\setminus T| + k$, then there is a set $S \subseteq V(G)$ of size at most $k$ such that $G-S$ has no odd-length $T$-path.
\end{claim}
\begin{claimproof}
Observe that since $G$ is bipartite, \bmp{the graph $H$ is bipartite as well}. By K\H{o}nig's theorem~\cite[Theorem 16.2]{Schrijver03}, $H$ has a vertex cover $X$ of size $|V(G)\setminus T| + k$. 
Let $S = (X \cap T) \cup \{u \mid \{u, u'\} \subseteq X\}$. Note that for each $u \in V(G) \setminus T$, at least one of $u \in X$ or $u' \in X$ must hold \bmp{to cover the edge~$uu'$, thereby} already accounting for $|V(G) \setminus T|$ vertices of the cover. It follows that $|S| \leq k$. We argue that $S$ hits all odd-length $T$-paths in~$G$. 

For the sake of contradiction suppose there is some odd-length $T$-path from $t_1 \in T$ to $t_2 \in T \setminus \{t_1\}$ in $G-S$. Let $P = (t_1,v_1,\dots,v_{2\ell},t_2)$ be a (not necessarily induced) shortest \bmp{odd $T$-path}. Note that neither of $t_1$ and $t_2$ belong to $X$ by construction of $S$. Furthermore $P \setminus \{t_1,t_2\} \subseteq V(G) \setminus T$, as otherwise we could construct a shorter odd-length $T$-path. Consider the vertices $V(P) \cup \{v_1',\dots,v_{2\ell}'\} \subseteq V(H)$. By construction of $S$, the vertex cover $X$ has exactly one of $v_i$ and $v_i'$ for each $i \in [2\ell]$. Observe that for it to be a vertex cover, $X$ must contain vertices of $P \setminus \{t_1,t_2\}$ and their copies in an alternating fashion \bmp{since for each edge~$v_i, v_{i+1}$ of~$P$, the graph~$H$ contains edges~$v_i v'_i, v_{i+1} v'_{i+1}, v_i v_{i+1}, v'_i v'_{i+1}$}. Without loss of generality, let~$v_1 \in X$. It follows that $v_{2\ell} \notin X$. But this contradicts that $X$ is a vertex cover as $v_{2\ell}t_2$ is not covered. We conclude that $S$ hits all odd-length $T$-paths in $G$.
\end{claimproof}

By Claim~\ref{claim:oddTpaths_to_matching}, \bmp{a maximum packing of $k$ of vertex-disjoint odd $T$-paths in $G$} implies a matching in $H$ of size $|V(G) \setminus T| + k$. By Claim~\ref{claim:bipartite}, we can create a set of size at most~$k$ that intersects all odd $T$-paths in the bipartite graph $G$. Clearly such a set has size at least $k$. It follows that if $G$ is bipartite, then the cardinality of a maximum packing of odd $T$-paths is equal to the minimum size of a vertex set which intersects all odd $T$-paths. (For completeness, we remark that this last property can also be derived from Menger's theorem: in a bipartite graph with bipartition into~$A \cup B$, a $T$-path is odd if and only if its endpoints belong to different partite sets, so a maximum packing of odd $T$-paths is equivalent to a maximum set of vertex-disjoint paths between~$A \cap T$ and~$B \cap T$.)
\end{proof}

By observing that odd $T$-paths in $G-v$ directly correspond to flowers with odd cycles pairwise intersecting at $v$ in $G$, Lemma~\ref{lem:oddflower+bipartite_packing_covering_equality} and Theorem~\ref{thm:maxflower+minmax_on_C} imply the following. 

\begin{lemma}\label{lemma:oct-2-essential}
There is a polynomial-time algorithm for \psdetectFor{2}{$\oct$}.
\end{lemma}

\subsection{Directed odd cycle transversal}
The $\doct$ problem corresponds to $\C$-deletion where $\forb(C)$ consists of all directed cycles of odd length. \bmp{Using Menger's theorem on an auxiliary graph, we can detect $3$-essential vertices for this problem.}

\begin{lemma} \label{lem:doct}
There is a polynomial-time algorithm for \psdetectFor{3}{$\doct$}.
\end{lemma}
\begin{proof}
Consider an instance $(D,k)$ of \psdetectFor{3}{$\doct$}.  Construct the so called label-extended graph (cf.~\cite[Section 3.1]{DBLP:conf/approx/AlevL17}) $D'$ (initially empty) as follows. For each $v \in V(D)$, add $v'$ and $v''$ to $V(D')$. For each directed arc $(u,v) \in E(D)$, add the arcs $(u',v'')$ and $(u'',v')$ to $E(D')$. For each vertex $v \in V(D)$, compute a minimum $(v',v'')$-separator $S_v$ in $D'$. Let $Q \subseteq V(D)$ be the set of vertices with $v \in Q$ if and only if $|S_v| \geq 2k+1$. We argue that $Q$ satisfies the output requirements~\ref{prop:c-essential-sup-det-double-guarantee-1} and~\ref{prop:c-essential-sup-det-double-guarantee-2}. 

First, suppose that $\opt_{\doct}(G) \leq k$. Consider a vertex $v \in Q$. Since $|S_v| \geq 2k+1$, by Menger's theorem there is a collection $P_1', \dots, P_{|S_v|}'$ of internally vertex-disjoint $(v',v'')$-paths in $D'$. Let $P_i = \{u \mid \{u',u''\} \cap V(P_i') \neq \emptyset \}$ and $\mathcal{P} = \bigcup_i \{P_i\}$. 

\begin{claim}
$D[P_i]$ contains an odd directed cycle for each $i \in [|S_v|]$. Furthermore, each $u \in V(D) \setminus \{v\}$ intersects at most two vertex sets of $\mathcal{P}$.
\end{claim}
\begin{claimproof}
\jjh{Each $(v',v'')$-path $P_i'$ in $D'$ corresponds to a directed odd walk $P_i$ in $D$, which contains a directed odd cycle. To see the second point, note that the paths in $D'$ are vertex-disjoint and since we created two copies of each vertex, it follows that each vertex of $V(D) \setminus \{v\}$ intersects at most two vertex sets of $\mathcal{P}$.}
\end{claimproof}
By the claim above, any solution to $\doct$ avoiding \bmp{$v \in Q$} has size at least $k+1$. It follows that each vertex in $Q$ belongs to every optimal solution and therefore Property~\ref{prop:c-essential-sup-det-double-guarantee-1} is satisfied.
Next, suppose that $\opt_{\doct}(G) = k$. We argue that $Q$ contains all 3-essential vertices. Consider an optimal solution $X$ a vertex $v \notin Q$. If $v \notin X$, then clearly $v$ is not 3-essential and there is nothing to show, so suppose that $v \in X$. Let $X' = \{u \in V(D) \mid \{u', u''\} \cap S_v \neq \emptyset\}$. Note that $|X'| \leq |S_v|$. Since any odd cycle containing $v$ corresponds to a $(v',v'')$-path in $D'$, and since $S_v$ is a $(v',v'')$-separator in $D'$, it follows that $(X \setminus \{v\}) \cup X'$ is a solution of size at most $k-1 + 2k+1 = 3k$ that avoids $v$. It follows that $v$ is not $3$-essential, therefore $Q$ contains all 3-essential vertices and Property~\ref{prop:c-essential-sup-det-double-guarantee-2} is satisfied.
\end{proof}

\bmp{We cannot use the approach based on computing a maximum $(v,\mathcal{F})$-flower for the \textsc{Chordal Deletion} problem; a simple reduction\footnote{Starting from an instance~$(G,(s_1,t_1),\ldots,(s_\ell,t_\ell))$ of \textsc{Disjoint Paths} satisfying~$s_it_i\notin E(G)$ for all~$i\in [\ell]$ (which is without loss of generality), insert a vertex~$v$ adjacent to~$A = \bigcup_{i=1}^\ell \{s_i, t_i\}$ and insert all edges between vertices in~$A$ except~$s_it_i$ for each~$i \in [\ell]$.} from \textsc{Disjoint Paths}~\cite{RobertsonS95b} shows that it is NP-hard to compute a maximum~$(v, \mathcal{F})$-flower when $\mathcal{F}$ is the set of chordless cycles of length at least four. In the next section, we will therefore use an approach based on the linear-programming relaxation to deal with \textsc{Chordal Deletion}.}

\section{Positive results via Linear Programming}\label{sec:positive_lp}

\bmp{Consider the following natural linear program for \CDeletion for hereditary graph classes~$\C$. The LP corresponding to an input graph~$G$ is defined on the variables~$(x_u)_{u \in V(G)}$, as follows.

\deflp{$\C$-Deletion LP}
{minimize~$\sum _{u \in V(G)} x_u.$}
{\begin{itemize}
    \item $\sum _{u \in V(H)} x_u \geq 1$ for each induced subgraph~$H$ of~$G$ isomorphic to a graph in~$\forb(\C)$,
    \item $0 \leq x_u \leq 1$ for each $u \in V(G)$.
\end{itemize}}

In the corresponding integer program, the constraint~$0 \leq x_u \leq 1$ is replaced by~$x_u \in \{0,1\}$. We say that a minimization LP has \emph{integrality gap} at most~$c$ for some~$c \in \mathbb{R}$ if the cost of an integer optimum is at most~$c$ times the cost of a fractional optimum. In general, the number of constraints in the \textsc{$\C$-Deletion LP} can be exponential in the size of the graph. Using the ellipsoid method (cf.~\cite{Schrijver03}), this can be handled using \ben{a} separation oracle: a polynomial-time algorithm that, given an assignment to the variables, outputs a violated \ben{constraint} if one exists. It is well-known (cf.~\cite[Thm. 5.10]{Schrijver03}) that linear programs with an exponential number of constraints can be solved in polynomial time using a polynomial-time separation oracle. 
To detect essential vertices, the integrality gap of a slightly extended LP will be crucial. We define the \textsc{$v$-Avoiding $\C$-Deletion LP} for a graph~$G$ and distinguished vertex~$v \in V(G)$ as the \textsc{$\C$-Deletion LP} with the additional constraint that~$x_v = 0$. Hence its integral solutions correspond to $\C$-modulators avoiding~$v$.}

\begin{theorem} \label{thm:lp:metathm}
Let $\C$ be a hereditary graph class such that for each graph~$G$ and~$v \in V(G)$ satisfying~$G - v \in \C$, the integrality gap of \textsc{$v$-Avoiding $\C$-Deletion} on~$G$ is at most~$c \in \mathbb{R}_{\geq 1}$. If there is a polynomial-time separation oracle for the \textsc{$\C$-Deletion LP}, then there is a polynomial-time algorithm for 
\psdetectFor{(c+1)}{\CDeletion}.
\end{theorem}
\begin{proof}
Given~$G$ and~$k$, the detection algorithm initializes an empty vertex set~$S$ and proceeds as follows. For each~$v \in V(G)$, it solves the \textsc{$v$-Avoiding $\C$-Deletion LP} on~$G$ in polynomial time using the ellipsoid method via the polynomial-time separation oracle. If the linear program has cost more than~$k$, we add~$v$ to~$S$. After having considered all vertices~$v \in V(G)$, the resulting set~$S$ is given as the output. 

To see that the output satisfies Property~\ref{prop:c-essential-sup-det-double-guarantee-1}, assume that~$\opt_\C(G) \leq k$ and consider some optimal $\C$-modulator~$X$ of size at most~$k$. If~$v \notin X$ then~$X$ induces an integer feasible solution to the \textsc{$\C$-Deletion LP} that satisfies the additional constraint $x_v = 0$, so that the cost of the \textsc{$v$-Avoiding $\C$-Deletion LP} is at most~$k$ and therefore~$v \notin S$. By contraposition, all vertices of~$S$ are indeed contained in some minimum $\C$-modulator, namely~$X$.

To see that the algorithm also satisfies Property~\ref{prop:c-essential-sup-det-double-guarantee-2}, assume~$\opt_\C(G) = k$ and let~$X$ be a minimum $\C$-modulator of size~$k$. We prove that~$S$ contains all $(c+1)$-essential vertices. Consider an arbitrary~$v \notin S$; we will argue it is not $(c+1)$-essential by exhibiting a $(c+1)$-approximate modulator avoiding~$v$. Since~$v \notin S$, the \textsc{$v$-Avoiding $\C$-Deletion LP} has a (fractional) solution~$\mathbf{x} = (x_u)_{u \in V(G)}$ of cost at most~$k$. If~$v \notin X$ then~$v$ was not $(c+1)$-essential, and there is nothing to show. So assume~$v \in X$.

Restricting the assignment~$\mathbf{x}$ to the vertices of the graph~$G' := G - (X \setminus \{v\})$ yields a feasible solution~$\mathbf{x}'$ for the \textsc{$v$-Avoiding $\C$-Deletion LP} on~$G'$, whose cost is at most that of~$\mathbf{x}$ and therefore at most~$k$. Note that~$G' - v$ equals~$G - X$ and therefore belongs to~$\C$. So by the precondition to the theorem, the integrality gap for \textsc{$v$-Avoiding $\C$-Deletion} on~$G'$ is at most~$c$. Hence the solution $\mathbf{x}'$ can be rounded to an integral solution~$X'$ on~$G'$ of cost at most~$c \cdot k$ and~$v \notin X'$ due to the constraint~$x_v = 0$. Since~$G - ((X \setminus \{v\}) \cup X') = G' - X' \in \C$, it follows that~$(X \setminus \{v\}) \cup X'$ is a $\C$-modulator of size at most~$c \cdot k + k = (c+1)k$ avoiding~$v$, which is therefore a~$(c+1)$-approximation. Hence~$v$ is not $(c+1)$-essential whenever~$v \notin S$.
\end{proof}

Using known results on covering versus packing for chordless cycles in near-chordal graphs, the approach above can be used to detect essential vertices for \textsc{Chordal Deletion}. \bmp{For the class of chordal graphs}, the corresponding set of forbidden induced subgraphs is the class \hole of all holes, i.e., induced chordless cycles of length at least four.

\begin{lemma}[{\cite[Lemma 1.3]{DBLP:journals/siamdm/JansenP18}}] \label{lem:chordal:packingcovering}
There is a polynomial-time algorithm that, given a graph $G$ and a vertex $v$ such that $G-v$ is chordal, outputs a \flower{$v$}{\hole} with $p$-petals and a set $S \subseteq V(G) \setminus \{v\}$ of size at most $12p$ such that $G-S$ is chordal.
\end{lemma}

Using this covering/packing statement, we can bound the relevant integrality gap and thereby detect essential vertices for \textsc{Chordal Deletion}.

\begin{lemma} \label{lem:chordal:detect}
There is a polynomial-time algorithm for \psdetectFor{13}{$\chordDel$}.
\end{lemma}
\begin{proof}
We first argue that the integrality gap \bmp{for \textsc{$v$-Avoiding Chordal Deletion} is bounded by~$12$ when~$G-v$ is chordal.} By Lemma~\ref{lem:chordal:packingcovering}, there is a value of~$p$ such that~$G$ contains both a $(v, \hole)$-flower~$\{C_1, \ldots, C_p\}$ with $p$-petals and a $v$-avoiding chordal modulator~$S$ of size at most~$12p$. Due to the constraint~$x_v = 0$, any fractional solution to the linear program has a cost of at least~$p$, since~$\sum _{u \in C_i} x_u \geq 1$ while~$x_v = 0$ and the holes \ben{only intersect at $v$}. At the same time, an integral solution has cost at most~$12p$ as~$S$ is such a solution. Hence the integrality gap is at most 12.

To invoke Theorem~\ref{thm:lp:metathm}, it suffices to argue that there is a polynomial-time separation oracle for the linear program. \bmp{Such a separation oracle is known (cf.~\cite[\S 10.1]{DBLP:journals/siamdm/JansenP18}); we repeat it here for completeness.} To test whether an assignment~$\mathbf{x} = (x_u)_{u \in V(G)}$ satisfies all constraints, it suffices to do the following. For each~$u \in V(G)$ we find the minimum total weight of any hole involving~$u$, as follows. For each pair of distinct non-adjacent~$p,q \in N_G(u)$ we use Dijkstra's algorithm to find the minimum weight~$W$ of a path~$P$ from~$p$ to~$q$ in the graph~$G - (N_G[u] \setminus \{p,q\})$ where the values of~$\mathbf{x}$ are used as the vertex weights. There is a hole of weight less than~$1$ through~$(p,u,q)$ \bmp{if and only if~$W + x_u < 1$}. Moreover, if~$W + x_u < 1$ then by extending path~$P$ with the vertex~$u$ we find a hole whose total weight is less than~$1$ and therefore a violated constraint. The fact that we remove all vertices of~$N_G[u]$ other than~$p$ and~$q$ ensures that the cycle we get in this way is induced, while the non-adjacency of~$p$ and~$q$ ensures it has at least four vertices. Hence after iterating over all choices of~$u,p,q$ we either find a violated constraint or conclude that the assignment is feasible. This shows that Theorem~\ref{thm:lp:metathm} is applicable and concludes the proof.
\end{proof}

\section{Consequences for Parameterized Algorithms} \label{sec:fpt}
In this section we show how the algorithms for \psdetectshort{$c$} from the previous section can be used to solve \CDeletion for various classes~$\C$, despite the fact that the detection algorithms only work when certain guarantees on~$k$ are met. The main theorem connecting the detection problem to solving \CDeletion is the following.

\begin{theorem}\label{thm:non-essentiality}
Let $\mathcal{A}$ be an algorithm that, given a graph $G$ and an integer $k$, runs in time~$f(k) \cdot |V(G)|^{\Oh(1)}$ for some non-decreasing function~$f$ and returns a minimum-size  $\C$-modulator if there is one of size at most~$k$. 
Let $\mathcal{B}$ be a polynomial-time algorithm for \psdetectFor{c}{\CDeletion}. Then there is an algorithm that, given a graph $G$, outputs a minimum-size $\C$-modulator in time $f(\ell) \cdot |V(G)|^{\Oh(1)}$, where~$\ell = \opt_\C(G) - |\mathcal{E}_c(G)|$ is the \emph{$c$-non-essentiality}.
\end{theorem}
\begin{proof}
First we describe the algorithm as follows. For each $0 \leq k \leq |V(G)|$, let~$S_k$ be the result of running~$\mathcal{B}$ on~$(G,k)$, let~$G_k := G - S_k$, and let~$b_k := k - |S_k|$.

Letting $L$ be the list of all triples $(G_k, S_k, b_k)$ sorted in \bmp{increasing} order by their third component~$b_k$, proceed as follows. For each $(G_k, S_k, b_k) \in L$, run~$\mathcal{A}$ on~$(G_k,b_k)$ to find a minimum $\C$-modulator~$S_{\mathcal{A}}$ of size at most~$b_k$, if one exists. If~$|S_{\mathcal{A}}| = b_k$, then output $S_{\mathcal{A}} \cup S_k$ as a minimum $\C$-modulator in~$G$.

To analyze the algorithm, we first argue it always outputs a solution. For the call with~$k^* = \opt_\C(G)$, both conditions of the detection problem are met. Hence by Property~\ref{prop:c-essential-sup-det-double-guarantee-1} the set~$S_{k^*}$ is contained in a minimum modulator in~$G$, so that ~$\opt_\C(G - S_{k^*}) = \opt_\C(G) - |S_{k^*}| = k^* - |S_{k^*}|$. Therefore graph~$G_{k^*} = G - S_{k^*}$ has a modulator of size at most~$b_{k^*} = \opt_\C(G - S_{k^*})$ and none which are smaller, so that~$\mathcal{A}$ correctly outputs a modulator of size~$b_{k^*}$. In turn, this causes the overall algorithm to terminate with a solution.

Having established that the algorithm outputs a solution, we proceed to show that it outputs a minimum-size modulator whenever it outputs a solution (which may be in an earlier iteration than for~$k^* = \opt_\C(G)$). Let~$k'$ be the value of~$k$ that is reached when the algorithm outputs a solution~$S_\mathcal{A} \cup S_{k'}$. Then we know:
\begin{enumerate}
    \item algorithm~$\mathcal{A}$ found a minimum-size modulator~$S_\mathcal{A}$ in~$G_{k'}$ of size at most~$b_{k'}$, and
    \item the set~$S_\mathcal{A} \cup S_{k'}$ is a modulator in~$G$, since~$S_\mathcal{A}$ is a modulator in~$G_{k'} = G - S_{k'}$, and therefore~$\opt_\C(G) \leq b_{k'} + |S_{k'}| = k'$.
\end{enumerate}

To see that the algorithm is correct, notice that, since~$\opt_\C(G) \leq k'$, the set $S_{k'}$ is contained in some minimum-size modulator for $G$ (Property~\ref{prop:c-essential-sup-det-double-guarantee-1} of $\mathcal{B}$). Hence~$\opt_\C(G_{k'}) = \opt_\C(G) - |S_{k'}|$. Since $\mathcal{A}$ outputs a minimum-size modulator if there is one of size at most~$b_k$, we have~$|S_\mathcal{A}| = \opt_\C(G) - |S_{k'}|$, so that $\mathcal{A}(G_{k'}) \cup S_{k'}$ is a feasible modulator of size~$\opt_\C(G)$ and therefore optimal.

Now we prove the desired running-time bound. First of all, notice that we can determine the list $L$ in polynomial time by running~$\mathcal{B}$ once for each value of~$k$ (which is at most~$|V(G)|$). By how we sorted $L$, we compute $\mathcal{A}(G_k, b_k)$ only when~$b_k \leq b_{k^*}$, as we argued above that if the algorithm has not already terminated, it does so after reaching~$k^* = \opt_\C(G)$. Hence the calls to algorithm~$\mathcal{A}$ are for values of the budget~$b_k$ which satisfy~$b_k \leq b_{k^*}$. We bound the latter, as follows.

Since~$k^* = \opt_\C(G)$, the set $S_{k^*}$ found by $\mathcal{B}$ is a superset of the set $\mathcal{E}_c(G)$ of all of the $c$-essential vertices in $G$ (Property~\ref{prop:c-essential-sup-det-double-guarantee-2}). This means that we have 
\[
    b_{k^*} = \opt_\C(G) - |S_{k^*}| \leq \opt_\C(G) - |\mathcal{E}_c(G)| = \ell,
\]
so the parameter of each call to~$\mathcal{A}$ is at most~$\ell$, giving the total time bound $f(\ell) \cdot |V(G)|^{\Oh(1)}$.
\end{proof}

Theorem~\ref{thm:fpt:main} now follows from Theorem~\ref{thm:non-essentiality} via the algorithms for \psdetectshort{$c$} given in the previous sections and the state-of-the-art algorithms for the natural parameterizations listed in Table~\ref{table:summary}. \bmp{Although the latter may be originally stated for the decision version, using self-reduction they can easily be adapted to output a minimum solution if there is one of size at most~$k$.}

\section{Hardness results}\label{sec:hardness}
Given the positive results we saw in Sections~\ref{sec:positive_packcover} and~\ref{sec:positive_lp}, it is natural to seek problems $\Pi$ for which \psdetect{c} is intractable. Here we show that \psdetectFor{c}{Dominating Set} is intractable for any $c \in \Oh(1)$ and then use  this as a starting point to prove similar results for \textsc{Hitting Set}, \textsc{Perfect Deletion}, \ben{and \textsc{Wheel-free Deletion}}.

\bmp{A dominating set is a vertex set whose closed neighborhood is the entire graph. The domination number of a graph is the size of a minimum dominating set.} The starting point for our reductions is the following result which states that it is $W[1]$-hard to solve \textsc{Dominating set} parameterized by solution size even on instances which have `solution-size gaps'.

\begin{lemma}[\cite{karthik2017parameterized}, cf.~{\cite[Thm.~4]{approxHardnessSurvey}}]\label{lemma:dom-set-gap-hard}
Let $F, f \colon \mathbb{N} \to \mathbb{N}$ be any computable functions. \bmp{Assuming $\mathsf{FPT} \neq \mathsf{W[1]}$, there does not exist an algorithm that, given a graph~$G$ and integer~$k$, runs in time~$f(k) \cdot |V(G)|^{\Oh(1)}$ and distinguishes between the following two cases:}
\begin{itemize}
    \item Completeness: $G$ has a dominating set of size $k$.
    \item Soundness: Every dominating set of $G$ is of size at least $k \cdot F(k)$.
\end{itemize}
\end{lemma}

All of our reductions in this section share a leitmotif. We start with a \emph{gap} instance $(G,k)$ of \textsc{Dominating Set} and map it to an instance $G'$ of \psdetect{c} (for appropriate $\Pi$) equipped with a distinguished vertex $\ben{v^*}$ with the following property: (1) if $G$ has domination number at most $k$, then no optimal solution in $G'$ contains $\ben{v^*}$; (2) if $G$ has domination number strictly greater than $c \cdot F(k)$ (for some appropriate $F$), then $\ben{v^*}$ is contained in every solution of size at most $c \cdot F(k)$ in $G'$. 
Thus our hardness results will follow by combining reductions of this kind with Lemma~\ref{lemma:dom-set-gap-hard}.

\bmp{\begin{lemma}\label{lemma:dom-set-reduction}
There is a polynomial-time algorithm~$R$ that, given a graph~$G$ and integer~$k$, outputs a graph~$R(G,k)$ containing a distinguished vertex~$\ben{v^*}$ such that:
\begin{itemize}
    \item if $G$ has dominating number at most~$k$, then the domination number of $R(G,k)$ is exactly~$k$ and every optimal dominating set avoids $\ben{v^*}$; 
    \item if $G$ has domination number strictly greater than $c\cdot (k+1)$ for some~$c \in \mathbb{R}_{\geq 1}$, then $R(G,k)$ has domination number $k + 1$ and the distinguished vertex $\ben{v^*}$ is contained in all $R(G,k)$-dominating sets of size at most $c \cdot (k+1)$.
\end{itemize}
\end{lemma}}
\begin{proof}
The graph $R(G,k)$ (see Figure~\ref{fig:reduction-dom-set}) is defined formally as follows:
\begin{itemize}
    \item \bmp{Initialize~$R(G,k)$ as the graph on vertex set~$\{v_i \mid v \in V(G), 0 \leq i \leq k\}$ with edges~$\{ v_i u_j \mid uv \in E(G), 0 \leq i, j \leq k\}$.
    \item For each~$i \in \{1, \dots, k\}$ insert an apex~$a_i$ which is adjacent to~$\{v_i \mid v \in V(G)\}$.
    \item Insert a vertex~$v^*$ which is adjacent to~$\{v_i \mid v \in V(G), 0 \leq i \leq k\}$.}
\end{itemize}
For our notational convenience, notice that the graph $R(G,k)$ contains $k+1$ isomorphic copies of $G$ denoted as $G_0, \dots, G_k$ where each $G_i$ is defined as the induced subgraph $R(G,k)[\{v_i \mid v \in V(G)\}]$ of $R(G,k)$.

\begin{figure}[t]
    \centering
\begin{tikzcd}
	&&&& {a_1} & {a_2} & \dots & {a_k} \\
	G & {} & {} & {G_0} & {G_1} & {G_2} & \dots & {G_k} \\
	&&&&& \ben{v^*}
	\arrow[squiggly, no head, from=2-5, to=2-6]
	\arrow[squiggly, no head, from=2-6, to=2-7]
	\arrow[squiggly, no head, from=2-7, to=2-8]
	\arrow[squiggly, no head, from=2-4, to=2-5]
	\arrow[Rightarrow, no head, from=1-5, to=2-5]
	\arrow[Rightarrow, no head, from=1-6, to=2-6]
	\arrow[Rightarrow, no head, from=1-8, to=2-8]
	\arrow[Rightarrow, no head, from=3-6, to=2-4]
	\arrow[Rightarrow, no head, from=3-6, to=2-5]
	\arrow[Rightarrow, no head, from=3-6, to=2-6]
	\arrow[Rightarrow, no head, from=3-6, to=2-8]
	\arrow["R", maps to, from=2-1, to=2-3]
\end{tikzcd}
    \caption{Reducing $R_k$ of $k-\ds$ to \psdetectFor{c}{$\ds$}. \bmp{Each} $G_i$ is a copy of $G$ and for $j \neq i$, $G_i$ is a twin to $G_j$ (denoted in the figure by squiggly lines: $\sim$) \ben{i.e. for each $x \in V(G)$ and any two distinct copies $x_i$ and $x_j$ of $x$ in $G_i$ and $G_j$, $x_i$ is a true twin to $x_j$. E}ach $a_i$ is an apex vertex to $G_i$ (i.e. $N(a_i) = V(G_i)$) and $\ben{v^*}$ is an apex to the whole graph except $a_1, \dots, a_k$ (i.e. $N(\ben{v^*}) = V(G_0 \uplus G_1 \dots \uplus G_k)$).}
    \label{fig:reduction-dom-set}
\end{figure}
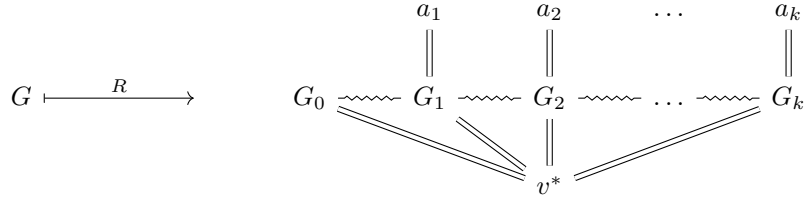

Notice that $R(G,k)$ has domination number at least $k$: since $\{a_1, \dots, a_k\}$ is an independent set and the neighborhoods of each $a_i$ are disjoint in $R(G,k)$, we have that any dominating set of $R(G,k)$ is forced to pick at least one vertex from $V(G_i) \cup \{a_i\}$ for each $1 \leq i \leq k$. In particular, this means that every dominating set that contains $\ben{v^*}$ must have cardinality at least $k+1$ (since $\ben{v^*}$ does not dominate $a_1, \dots, a_k$).

Now, to prove that $R(G,k)$ has the desired properties, start by supposing that $G$ has a dominating set $\{\ben{s_1}, \dots s_\ell\}$ with $\ell$ at most $k$. 
We claim that, by `spreading-out' this dominating set over $\jjh{G_1} , \dots, G_\ell$ (i.e. pick the copy of $s_i$ in $G_i$ for each $i$) we obtain a vertex-subset $S'$ of $R(G,k)$ which dominates $\ben{v^*}$, $\{a_1, \dots, a_\ell\}$ and every vertex of $G_0, \dots, G_\ell$. The first two observations are immediate and the last follows since each $G_i$ is a `twin' of $G_j$ in $R(G,k)$: any vertex $y$ in $G$ is dominated by at least one vertex, $s_i$ of $S$ and, by the construction of $R(G,k)$, the copy of $s_i$ in $G_i$ dominates all of the copies $y_0, \dots, y_\ell$ of $y$ in $G_0, \dots, G_\ell$. It follows immediately from the definition of $R(G,k)$ that $S' \cup \{a_i: \ell \jjh{ < } i \leq k\}$ is a dominating set of $R(G,k)$ of size $k$ (which is smallest-possible by our previous arguments). Thus, by our previous discussion, all minimum-size dominating sets of $R(G,k)$ avoid $\ben{v^*}$.

Now suppose there exists a constant $c \geq 1$ \bmp{such that~$G$ has domination number greater than~$c \cdot (k+1)$.}
Take any dominating set $S$ of $R(G,k)$ not containing $\ben{v^*}$. Since $G_0$ is not adjacent to any one of the vertices $a_1, \dots, a_k$ and $\ben{v^*} \not \in S$, every vertex of $G_0$ must be dominated by some vertex of $G_0 \uplus \dots \uplus G_k$ in $R(G,k)$. Thus the obvious projection of $S \setminus \{a_1, \dots, a_k, \ben{v^*}\}$ onto $G$ is a dominating set for $G$. Hence, if $S$ avoids $\ben{v^*}$, then $|S| > c \cdot (k+1)$ by our assumption on the domination number of $G$. In other words, every dominating set of size at most $c \cdot (k+1)$ must contain $\ben{v^*}$. This, combined with the fact that any dominating set of $R(G,k)$ is forced to pick at least one vertex from $V(G_i) \cup \{a_i\}$ for each $1 \leq i \leq k$ (as we observe earlier), implies that $\{\ben{v^*}, a_1, \dots, a_k\}$ is a minimum-size dominating set of $R(G,k)$ (where each $a_i$ dominates itself and $\ben{v^*}$ dominates everything else) \bmp{which has size~$k+1$}.
\end{proof}

Lemma~\ref{lemma:dom-set-gap-hard} combined with the reduction provided by Lemma~\ref{lemma:dom-set-reduction} yields the following.

\begin{theorem}\label{thm:dom-set-c-essential-detection-hardness}
Unless $\fpt = W[1]$, there is no $\fpt$-time algorithm for \psdetectFor{c}{$\ds$} parameterized by $k$ for \ben{any $c \in \mathbb{R}_{\geq 1}$.}
\end{theorem}
\begin{proof}
\bmp{Suppose such an algorithm~$\mathcal{A}$ exists for~$c$; we will use it with Lemma~\ref{lemma:dom-set-gap-hard} to show~$\fpt = W[1]$ for the function~$F(k) = c(k+1)$.

Given an input~$(G,k)$ in which the goal is to distinguish between the completeness and soundness cases, the algorithm proceeds as follows.} Using the reduction $R$ of Lemma~\ref{lemma:dom-set-reduction}, consider the graph $R(G,k)$ and let $S$ be the output of an algorithm for \psdetectFor{c}{$\ds$} on the pair $(R(G,k), k+1)$\bmp{; note that the solution size for which we ask is~$k+1$ rather than~$k$.} Without loss of generality we may assume~$k\geq 2$, as the distinction can trivially be made otherwise. \bmp{We will show that in the completeness case we have~$\ben{v^*} \notin S$, while in the soundness case we have~$\ben{v^*} \in S$, which allows us to distinguish between these cases by checking whether~$\ben{v^*}$ belongs to the output of~$\mathcal{A}(R(G,k),k+1)$.}

For the completeness case, suppose $G$ has domination number at most $k$. Then, by Lemma~\ref{lemma:dom-set-reduction}, so does $R(G,k)$. This means that Property~\ref{prop:c-essential-sup-det-double-guarantee-1} holds for the call to~$\mathcal{A}(R(G,k),k+1)$, so that there is some optimal solution $S'$ of size $k$ which contains $S$ and hence we have $\ben{v^*} \not \in S$ by Lemma~\ref{lemma:dom-set-reduction}. 

For the soundness case, suppose $G$ has domination number at least~$k \cdot F(k) = c(k+1)k > c(k+1)$ (we use~$k\geq 2$ here). Then by Lemma~\ref{lemma:dom-set-reduction}, graph $R(G,k)$ has domination number $k+1$ and $\ben{v^*}$ is contained in all its dominating sets of size at most~$c (k+1)$. In other words: $\ben{v^*}$ is $c$-essential in $R(G,k)$. Consequently, $\ben{v^*} \in S$ by Property~\ref{prop:c-essential-sup-det-double-guarantee-2} since the argument~$k+1$ we supplied to~$\mathcal{A}$ coincides with the optimum in~$R(G,k)$ in this case.

\bmp{If~$\mathcal{A}$ runs in FPT-time, then the overall procedure runs in FPT-time which implies~$\fpt = W[1]$ by Lemma~\ref{lemma:dom-set-gap-hard}.}
\end{proof}

Consider the \emph{closed-neighborhood mapping} $S$: i.e. the standard polynomial-time reduction from \textsc{Dominating Set} to \textsc{Hitting Set} which maps each graph $G$ to the hypergraph
\begin{equation}\label{eqn:open-set-mapping}S(G) := \Bigl ( V(G), \{N(x) \cup \{x\} : x \in \jjh{V(G)}\} \Bigr).\end{equation}
\bmp{The following observation captures the relation between dominating sets of~$G$ and hitting sets of~$S(G)$.}

\begin{observation} \label{obs:domset:hittingset}
Let~$G$ be a graph and~$X \subseteq V(G)$. Then~$X$ is a dominating set of~$G$ if and only if~$X$ is a hitting set of $S(G)$.
\end{observation}

\bmp{Note that, given the reduction $R$ of Lemma~\ref{lemma:dom-set-reduction}, the composite mapping $S \circ R$ relates $c$-essentiality to gaps in solution quality in much the same way as $R$ did.}

Now consider the parameter-preserving polynomial-time parameterized reduction $P$ (due to Heggernes et al.~\cite[Thm.~1]{HeggernesHJKV13}) taking instances $(H,k)$ of \textsc{Hitting Set} (parameterized by solution size) to \textsc{Perfect Deletion} (also parameterized by solution size) defined as follows. For each hyperedge $e = \{x_1, \dots, x_\ell\}$ in $H$, create $\ell + 1$ new vertices $c_{e,1}, \dots, c_{e,\jjh{\ell+1}}$ -- called \emph{gadget vertices} -- and an odd cycle \[C_e = \bigl (e \cup \{c_{e,1}, \dots, c_{e,\jjh{\ell+1}}\}, \; \{c_{e,1}c_{e,\jjh{\ell+1}}\} \cup \{x_ic_{e,j} : x_i \in e, j \in \{i, i+1\}\} \bigr)\] 
so that we may define $P(H)$ as the graph obtained by adding all possible edges between gadget-vertices of different cycles; formally we have
\begin{equation}\label{eqn:bart-reduction} 
    P(H) := \Bigl(\bigcup_{e \in E(H)} V(C_e), 
                  \biguplus_{e \in E(H)} E(C_e) 
                    \jjh{\cup \{c_{e,i}c_{e',j} \mid e \neq e' \in E(H), i,j \in [\ell+1]\}}
            \Bigr ).
\end{equation}

\begin{lemma}[{\cite[Thm.~1 and Claims 1--3]{HeggernesHJKV13}}]\label{lemma:bart-reduction-properties}
Given any \textsc{Hitting Set} instance $(H,k)$ with each hyperedge of size at least $2$, the graph $P(H)$ defined in Equation~\eqref{eqn:bart-reduction} satisfies the following properties: 
\begin{enumerate}[label=\textbf{P\arabic*}]
    \item given any vertex-subset $X$ of $P(H)$ such that $P(H) - X$ is perfect; if there is an $x \in X$ such that $x \not \in V(P(H)) \cap V(H)$, then $x$ lies on exactly one odd-hole $C_e$, and, for any vertex $y \in V(C_e) \cap V(H)$ we have that $P(H) - (X \cup \{y\}) \setminus \{x\}$ is perfect, \label{prop:bart-reduction-1}
    \item for any $X \subseteq V(H)$, the set $X$ is a hitting set in $H$ if and only if $P(H) - X$ is perfect. \label{prop:bart-reduction-2}
\end{enumerate}
\end{lemma}

Lemma~\ref{lemma:bart-reduction-properties} ensures that we can chain the reductions $R$, $S$ and $P$ (found respectively in Lemma~\ref{lemma:dom-set-reduction} and Equations~\eqref{eqn:open-set-mapping} and~\eqref{eqn:bart-reduction}) to obtain a polynomial-time reduction from \textsc{Dominating Set} to the detection problem \psdetectFor{c}{$\perfectDel$} with sufficient guarantees to be able to then infer the intractability of the latter problem using Lemma~\ref{lemma:dom-set-gap-hard}.

\begin{theorem}\label{thm:perfect-del-hard}
Unless $\fpt = W[1]$, there is no $\fpt$-time algorithm for \psdetectFor{c}{$\perfectDel$} parameterized by $k$ for \bmp{any $c \geq 1$.}
\end{theorem}
\begin{proof}
\bmp{Fix any instance $(G,k)$ with~$k \geq 2$ of \textsc{Dominating Set} which has domination number at most $k$ or at least $ck(k+1) > c(k+1)$.} Then, by Observation~\ref{obs:domset:hittingset}, the hypergraph $S(R(G,k))$ has a distinguished vertex $\ben{v^*}$ such that following properties are satisfied: 
\begin{enumerate}[label=\textbf{C\arabic*}]
    \item if~$G$ has domination number at most~$k$, then $S(R(G,k))$ has hitting set number $k$ and all of its minimum hitting sets avoid $\ben{v^*}$; \label{pd:case1}
    \item if $G$ has domination number strictly greater than $c(k+1)$, then $S(R(G,k))$ has hitting set number $k + 1$ and its distinguished vertex $\ben{v^*}$ is contained in every solution of size at most $c(k+1)$.\label{pd:case2}
\end{enumerate}
Let $Q$ be the set returned by any algorithm for \psdetectFor{c}{$\perfectDel$} on input $\bigl( (P \circ S \circ R)(G,k), k+1 \bigr )$. Notice that, by Lemma~\ref{lemma:bart-reduction-properties}, in both Case~\ref{pd:case1} and Case~\ref{pd:case2} we have that the optimum for $\perfectDel$ on $(P \circ S \circ R)(G,k)$ coincides with the optimum for $\hittingset$ on $S(R(G,k))$. 

With this observation in mind, consider Case~\ref{pd:case1}. Here we have that $Q$ is contained in an optimum solution $Q'$ of size $k$ (by Property~\ref{prop:c-essential-sup-det-double-guarantee-1}). However, by the second point in Lemma~\ref{lemma:bart-reduction-properties} (and since all size-$k$ hitting sets of $S(R(G,k))$ avoid $\ben{v^*}$) we deduce that any such $Q'$ (and hence $Q$) cannot contain $\ben{v^*}$ (which is a vertex of $S(R(G,k))$ and hence also a vertex of $(P \circ S \circ R)(G,k)$).

For Case~\ref{pd:case2}, consider any perfect deletion set $X$ in $(P \circ S \circ R)(G,k)$ of size at most $c(k+1)$. If we can show that $X$ must always contain $\ben{v^*}$, then we are done since it would imply that $\ben{v^*}$ is $c$-essential which, by Property~\ref{prop:c-essential-sup-det-double-guarantee-2} would allows us to decide whether $G$ has domination number either $k$ or greater than $c(k+1)$ simply by checking whether $\ben{v^*}$ is in $Q$ or not. So, seeking a contradiction, suppose $\ben{v^*} \not \in X$. By Property~\ref{prop:bart-reduction-1} of Lemma~\ref{eqn:bart-reduction} and since each hyperedge $S(R(G,k))$ has size at least $2$, we can find a subset $Y$ of the vertices of $S(R(G,k)) - \ben{v^*}$ of size at most $|X|$ such that $(P \circ S \circ R)(G,k) - Y$ is perfect. But then, by Property~\ref{prop:bart-reduction-2} of Lemma~\ref{lemma:bart-reduction-properties}, we have that $Y$ is a hitting set of size at most $|X| = c(k+1) \leq ck(k+1)$ in $S(R(G,k))$ which avoids $\ben{v^*}$. However, this is a contradiction since Observation~\ref{obs:domset:hittingset} guarantees that $\ben{v^*}$ is contained in every hitting set for $S(R(G,k))$ of size at most $ck(k+1)$.
\end{proof}

In a similar vein to Heggernes et al.'s~\cite[Thm.~1]{HeggernesHJKV13} reduction $P$ (Equation~\ref{eqn:bart-reduction})) one can define a reduction $L$ (Figure~\ref{fig:wheel-reduction}) from \textsc{Hitting Set} to \textsc{Wheel-free Edge-Deletion}. This is due to Lokshtanov~\cite{Lokshtanov-wheel-free-2008} and it is defined as follows. Given an instance $(H, k)$ of \textsc{Hitting Set}, for each hyperedge $e \in E(H)$ with $e = \{x_1, \dots, x_\ell\}$, construct the wheel $W_e$ with $3\ell$ spokes and, picking some cyclic order of its rim-edges (i.e. the edges of the $3\ell$-cycle in the wheel) denote them as \[e_{x_1}, e_1', e_1'', e_{x_2}, e_2', e_2'',\dots, e_{x_\ell}, e_\ell', e_\ell''.\] 
We say that the edges $e_{x_1}, e_{x_2}, \dots e_{x_3}$ \emph{correspond} to the vertices $x_1, x_2, \dots, x_\ell$ respectively. Then, define the equivalence relation $\sim$ on the edges of the disjoint union of all such wheels so that any two edges that correspond to the same vertex become equivalent under $\sim$ (stating this formally, given any two distinct hyperedges $e, f$ in $H$, for all vertices $x \in e \cap f$, set $e_x \sim f_x$). Finally define $L(H)$ (see Figure~\ref{fig:wheel-reduction}) to be the quotient 
\begin{equation}\label{eqn:wheel-reduction}
    L(H) := \biguplus_{e \in E(H)} W_e /_{\sim}.
\end{equation}

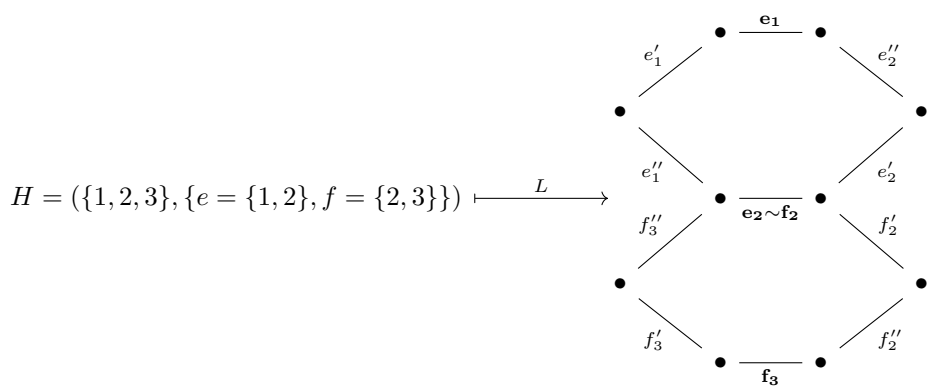
\begin{figure}
    \centering
\begin{tikzcd}
	&&& \bullet & \bullet \\
	&& \bullet && {} & \bullet \\
	{H= (\{1,2,3\},\{e = \{1,2\}, f =\{2,3\}\})} && {} & \bullet & \bullet \\
	&& \bullet &&& \bullet \\
	&&& \bullet & \bullet
	\arrow["L", maps to, from=3-1, to=3-3]
	\arrow["{\mathbf{e_2\sim f_2}}"', no head, from=3-4, to=3-5]
	\arrow["{e_2'}"', no head, from=3-5, to=2-6]
	\arrow["{e_2''}", no head, from=1-5, to=2-6]
	\arrow["{\mathbf{e_1}}"', no head, from=1-5, to=1-4]
	\arrow["{e_1'}"', no head, from=1-4, to=2-3]
	\arrow["{e_1''}"', no head, from=2-3, to=3-4]
	\arrow["{f_3''}"', no head, from=3-4, to=4-3]
	\arrow["{f_3'}"', no head, from=4-3, to=5-4]
	\arrow["{\mathbf{f_3}}"', no head, from=5-4, to=5-5]
	\arrow["{f_2''}"', no head, from=5-5, to=4-6]
	\arrow["{f_2'}", no head, from=3-5, to=4-6]
\end{tikzcd}
    \caption{An example of the reduction $L$ (we omit the centers of the two wheels above to improve readability). The edge marked with \textbf{bold} labels are those which correspond to vertices $1,2$ and $3$ of the Set-Cover instance $H$.}
    \label{fig:wheel-reduction}
\end{figure}

Given Lokshtanov's parameter-preserving reduction~\cite{Lokshtanov-wheel-free-2008} $L$ from \textsc{Hitting Set} to Wheel-free Edge-Deletion, we can argue in similar vein to Theorem~\ref{thm:perfect-del-hard} to show the intractability of \psdetectFor{c}{$\wheelfreeedge$}. 

\ben{Since we have so far focused on vertex-deletion problems, rather than proving this for the edge-deletion case, we will focus on showing the intractability of \psdetectFor{c}{$\wheelfree$}. To that end, we need the following slight modification of Lokshtanov's reduction. Given any hypergraph $H$, and wheels $W_e$ for each hyperedge $e$ of $H$ (as defined earlier), make a new wheel $W'_e$ from $W_e$ by subdividing each one of its edges and adding appropriate spokes. We shall say that a vertex of $W'_e$ which \emph{subdivides} an edge $e_x$ of $W_e$ \emph{corresponds} to a vertex $x$ in $H$ if and only if the edge $e_x$ corresponds to the vertex $x$. Finally we define $\Lambda(H)$ as the quotient $\Lambda(H) := \biguplus_{e \in E(H)} W'_e /_{\sim_V}$ where $\sim_V$ identifies those \emph{vertices} of $W'_e$ which correspond to the same vertex in $H$. } 

\begin{lemma}\label{lemma:wheel-help}
\ben{For any $\hittingset$ instance $H$ and any integer $k \leq |V(H)|$ we have that $H$ has a hitting set of size $k$ if and only if $\Lambda(H)$ has a vertex-modulator to wheel-freeness $X$ of size $k$ such that every vertex of $X$ corresponds to a vertex in $H$.}
\end{lemma}
\begin{proof}
\ben{By Lokshtanov's reduction~\cite{Lokshtanov-wheel-free-2008}, we have for all $k$ that: $H$ has a hitting set of size $k$ if and only if $L(H)$ has an edge-modulator to wheel-freeness $F$ of size $k$ such that every edge of $F$ corresponds to a vertex in $H$. Note that, given any such $F$, we can turn this into a vertex-modulator for wheel-freeness (of the same size as $F$ and whose elements also all correspond to vertices of $H$) on $\Lambda(H)$ by mapping each edge in $F$ to its subdividing vertex. 
Furthermore, since (by construction) the only vertices of $\Lambda(H)$ that can be in more than one wheel are those in $V(\Lambda(H)) \setminus V(L(H))$, every solution of $\Lambda(H)$ can be assumed to consist solely of vertices in $V(\Lambda(H)) \setminus V(L(H))$. Thus, if $\Lambda(H)$ has a solution of size $k \leq |V(H)|$, then we can turn it into a solution of size at most $k$ whose elements all correspond to vertices of $H$ (i.e. a solution in $H$). 
}
\end{proof}

\begin{theorem}\label{thm:wheel-free-hard}
Unless $\fpt = W[1]$, there is no $\fpt$-time algorithm for \psdetectFor{c}{$\wheelfree$} parameterized by $k$ for \bmp{any} $c$.
\end{theorem}
\begin{proof}
\bmp{Once again, starting with any instance $(G,k)$ with~$k\geq 2$ of \textsc{Dominating Set} which has domination number either $k$ or at least~$ck(k+1)$, we apply the closed-neighborhood mapping $S$} (Equation~\eqref{eqn:open-set-mapping}) so that we may proceed by case distinction on the same cases~\ref{pd:case1} and~\ref{pd:case2} as we did in the proof of Theorem~\ref{thm:perfect-del-hard}. 

Notice that, for any vertex $x$ in $(\Lambda \circ S \circ R)(G,k)$ which does not correspond to any vertex of $S(R(G,k))$, $x$ is contained in \emph{exactly one} wheel of $(\Lambda \circ S \circ R)(G,k)$. This, combined with Lemma~\ref{lemma:wheel-help}, implies that the optimum for $\wheelfree$ on $(\Lambda \circ S \circ R)(G,k)$ always coincides with the optimum for $\hittingset$ on $S(R(G,k))$. 

Let $Q$ be the set returned by any algorithm for \psdetectFor{c}{$\wheelfree$} on input $(\Lambda \circ S \circ R)(G,k)$. 

In Case~\ref{pd:case1}, since $(G,k)$ is a yes-instance, $S(R(G,k))$ has a minimum hitting set of size $k$ and every such minimum hitting set avoids the distinguished vertex $\ben{v^*}$. Thus, by Lemma~\ref{lemma:wheel-help}, we have that $(\Lambda \circ S \circ R)(G,k)$ has wheel-free deletion number $k$ and every witness $X$ of this fact satisfies the following property: no vertex in $X$ corresponds to $\ben{v^*}$ in any wheel of $(\Lambda \circ S \circ R)(G,k)$. Consequently, since the returned set $Q$ satisfies Property~\ref{prop:c-essential-sup-det-double-guarantee-1}, $Q$ does not contain any edge corresponding to $\ben{v^*}$ either.

In Case~\ref{pd:case2}, since $G$ has domination number strictly greater than $ck(k+1)$, $S(R(G,k))$ has hitting set number $k + 1$ and its distinguished vertex $\ben{v^*}$ is contained in every solution of size at most $c(k+1)$. Now take any wheel-free deletion set $X$ of size at most $c(k+1) = c \cdot \opt\bigl ( (\Lambda \circ S \circ R)(G,k) \bigr )$ (recall that the optimum for $\wheelfree$ on $(\Lambda \circ S \circ R)(G,k)$ always coincides with the optimum for $\hittingset$ on $S(R(G,k))$). Now, seeking a contradiction, suppose that no element of $X$ corresponds to $\ben{v^*}$. 
Since any element $x$ in $X$ is contained in precisely one wheel whenever it does not correspond to any vertex of $S(R(G,k))$, we can find always find a modulator $Y$ of size at most that of $X$ whose elements all correspond to vertices of $S(R(G,k))$. Furthermore, since we can assume that $G$ contains no isolated vertices, we can also assume that all hyperedges of $S(R(G,k))$ have size at least $2$. Thus every wheel in $(L \circ S \circ R)(G,k)$ has at least two vertices corresponding to vertices of $S(R(G,k))$ and hence we can further assume that no vertex in $Y$ corresponds to $\ben{v^*}$ (if it does, we can modify $Y$). But then the vertices corresponding to those of $Y$ form a hitting set of $S(R(G,k))$ which:  
\begin{itemize}
    \item has size at most $|Y| \leq |X| \leq c(k+1)$ and 
    \item avoids $\ben{v^*}$,
\end{itemize}
contradicting the fact that $\ben{v^*}$ is contained in every solution of size at most $c(k+1)$ in $S(R(G,k))$. Thus we have that $\ben{v^*} \in Q$ if and only if $\ben{v^*}$ is $c(k+1)$-essential in $S(R(G,k))$ \bmp{which is enough to distinguish whether $G$ has domination number at most $k$ or at least $ck(k+1)$.}
\end{proof}

\section{Connections to Perturbation Resilience} \label{sec:stability}
In the area of perturbation resilience~\cite{makarychev_makarychev_2021} there is a notion of so-called \textit{$c$-perturbation resilient instances} to optimization problems. Roughly, these are instances $G$ in which there is a \textit{unique} optimal solution $S$ which far outperforms (by a factor of $c$) every other solution $S'$ in $G$. More formally, for a vertex-weighted graph minimization problem $\Pi$ whose solution is a vertex-subset, we say that an instance $(G, w \colon V(G) \to \mathbb{N})$ with a \emph{unique} optimal solution $S$ is \emph{$c$-perturbation resilient} if for any \bmp{perturbed weight function $w'$ satisfying~$w(v) \leq w'(v) \leq c \cdot w(v)$ for all~$v \in V(G)$,} the instance $(G, w'(x))$ has a unique optimal solution and furthermore this solution is $S$. (Of course, one can define an analogous notion for maximization problems as well and what follows in this section applies to both.)

Classes of $c$-perturbation resilient instances have been shown to be `islands of tractability' where many intractable problems become polynomial-time-solvable~\cite{makarychev_makarychev_2021} and the suggestive intuition behind this fact is that perturbation resilient instances have unique optima which `stand out' and are `obvious' in some sense. 
Viewing stability through the lens of parameterized complexity, it is natural to ask whether one can quantify in an algorithmically useful way how distant an input is from being stable. The following proposition supports our claim that the \emph{non-essentiality} (recall Theorem~\ref{thm:non-essentiality}) is a good such measure since on $(>c)$-stable inputs, the $c$-non-essentiality is smallest possible (namely it is $0$).

\begin{proposition}\label{prop:stability-essentiality}
Given constants $c' > c \geq 1$, if $G$ is a $c'$-stable input to a graph optimization problem $\Pi$ whose solutions are vertex-subsets, then $\mathcal{E}^\Pi_c(G)$ is the unique optimum of $G$.
\end{proposition}
\begin{proof}
Consider the unique optimum $S$ for $\Pi$ on $G$. If $S'$ is a $c$-approximation for $\Pi$ on $G$, then we must have $S' = S$ (since otherwise $G$ would not be $c'$-stable). As a consequence we know that every vertex of $S$ must be $c$-essential and hence we have 
$S \subseteq \mathcal{E}_c(G) \subseteq S$.
\end{proof}

Proposition~\ref{prop:stability-essentiality} and Theorem~\ref{thm:non-essentiality} allow us to immediately deduce that any algorithm for \psdetectshort{$c$} solves all $(>c)$-perturbation resilient instances exactly.

\begin{corollary}\label{corollary:algorithm_stability_and_essentiality}
Given a \bmp{minimization} problem $\Pi$, any algorithm for \psdetect{c} solves $(>c)$-stable instances exactly.
\end{corollary}

\section{Conclusion}\label{sec:conclusion}

We introduced the notion of $c$-essential vertices for vertex-subset minimization problems on graphs, to formalize the idea that a vertex belongs to all \emph{reasonable} solutions. Using a variety of approaches centered around the theme of covering/packing duality, we gave polynomial-time algorithms that detect a subset of an optimal solution containing all $c$-essential vertices, which decreased the search space of parameterized algorithms from exponential in the size of the solution, to exponential in the number of non-essential vertices in the solution.

Throughout the paper we have restricted ourselves to working with unweighted problems. However, many of the same ideas can be applied in the setting where each vertex has a positive integer weight of magnitude~$\Oh(|V(G)|^{\Oh(1)})$ and we search for a minimum-weight solution. Since integral vertex weights can be simulated for many problems by making twin-copies of a vertex, our approach can be extended to \textsc{Weighted Vertex Cover}, \textsc{Weighted Odd Cycle Transversal}, and \textsc{Weighted Chordal Deletion}.

Our results shed a new light on which instances of NP-hard problems can be solved efficiently. FPT algorithms for parameterizations by solution size show that instances are easy when their optimal solutions are small. Theorem~\ref{thm:fpt:main} refines this view: it shows that instances with large optimal solutions can still be easy, as long as only a small number of vertices in the optimum is not $c$-essential.

We remark that there is an alternative route to algorithms for \psdetectshort{$c$}, which is applicable to \CDeletion problems which admit a constant-factor approximation. If there is a polynomial-time algorithm that, given a graph~$G$ and \bmp{vertex~$v$}, outputs a $c$-approximation for the problem of finding a minimum-size $\C$-modulator avoiding~$v$, it can be used for \psdetectshort{$c$}. A valid output~$S$ for the detection problem with input~$(G,k)$ is obtained by letting~$S$ contain all vertices for which the approximation algorithm outputs a $v$-avoiding modulator of size more than~$c \cdot k$. Using this approach (cf.~\cite{FominLMS12}) one can solve \psdetectFor{\max _{F \in \mathcal{F}} |V(F)|^{\Oh(1)}} {\textsc{$\mathcal{F}$-Minor-Free Deletion}} for any finite family~$\mathcal{F}$ containing a planar graph. As the results for problems for which no constant-factor approximation exists are more interesting, we focused on those.

Our work opens up several questions for future work. Is the integrality gap for \textsc{$v$-Avoiding Planar Vertex Deletion} constant, when~$G-v$ is planar? Can \psdetectFor{\Oh(1)}{\textsc{Planar Vertex Deletion}} be solved in polynomial time? Can \psdetectFor{2}{Chordal Deletion} be solved in polynomial time? Can the constant~$c$ for which we can detect $c$-essential vertices be lowered below~$2$?

Considering a broader horizon, it would be interesting to investigate whether there are less restrictive notions than $c$-essentiality which can be used as the basis for guaranteed search-space reduction.

\bibliography{biblio}
\bibliographystyle{plainurl}

\clearpage

\end{document}